\documentclass[11pt]{article}

\renewcommand{\d}{\operatorname{d}\!}

\newcommand{\Cov}{\operatorname{Cov}}

\usepackage{graphicx}
\usepackage{cite}
\usepackage{color}
\usepackage{amsfonts}
\usepackage{amsthm}
\usepackage{amsmath}
\usepackage{amssymb}
\usepackage{subfig}
\usepackage{fixltx2e}

\newcommand{\dLSI}{\delta_{\mathsf{LSI}} }

\usepackage{pseudocode}
\usepackage[]{fullpage}
\usepackage{enumerate}
\newtheorem{theorem}{Theorem}

\newtheorem{remark}{Remark}

\newtheorem{corollary}{Corollary}

\newtheorem{proposition}{Proposition}

\newcommand{\EE}{\mathbb{E}}

\title{Links between the  Logarithmic Sobolev Inequality and the  convolution  inequalities for Entropy and Fisher Information}

\author{Thomas~A.~Courtade \\ Department of Electrical Engineering and Computer Sciences\\University of California, Berkeley}

\begin{document}

\maketitle

\begin{abstract}
Relative to the  Gaussian measure on $\mathbb{R}^d$, entropy and Fisher information are famously related via Gross' logarithmic Sobolev inequality (LSI).  These same functionals also separately satisfy  convolution inequalities, as proved by Stam.  We establish a dimension-free inequality that interpolates among these relations.  %
Several interesting corollaries follow:  (i) the deficit in the LSI satisfies a convolution inequality itself; (ii) the deficit in the LSI controls convergence in the entropic and Fisher information central limit theorems; and (iii) the LSI is stable with respect to  {HWI jumps} (i.e., a jump in any of the convolution inequalities associated to the HWI functionals).

Another  consequence is that the convolution inequalities for Fisher information and entropy powers are reversible in general, up to a factor depending on the \emph{Stam defect}.  An improved form of Nelson's hypercontractivity estimate also follows.  Finally, we speculate on the possibility of an analogous reverse Brunn-Minkowski inequality and a related upper bound on  surface area associated to Minkowski sums.

\end{abstract}
 
 \section{Introduction}
 For a random vector $X$ on $\mathbb{R}^d$ with absolutely continuous density $f$, the entropy of $X$ is given by
 \begin{align}
h(X) = -\int f \log f,
\end{align}
provided the integral exists, and the entropy power of $X$ is defined according to 
\begin{align}
N(X) = \frac{1}{2\pi e} e^{\tfrac{2}{d} h(X)}.
\end{align}
The Fisher information of $X$ is defined by
\begin{align}
J(X) = \int f \left| \nabla  \log f  \right|^2, 
 \end{align}
with $J(X)=\infty$ if the integral does not exist. 

The entropy and Fisher information functionals play a fundamental role in information theory and related fields, and enjoy many useful properties.  Standing out among these properties is their behavior under convolution of densities.  In particular, Stam \cite{stam1959some} and Blachman \cite{blachman1965convolution} proved that if $X,Y$ are independent random vectors on $\mathbb{R}^d$, then 
\begin{align}
N(X+Y) \geq N(X)+N(Y)\label{EPI}
\end{align}
and
\begin{align}
\frac{1}{J(X+Y)}\geq \frac{1}{J(X)}+\frac{1}{J(Y)}.\label{FII}
\end{align}

These inequalities may also be stated in terms of relative entropies and Fisher informations, which will be useful for our purposes.  Toward this end, the entropy of $X$ relative to the standard normal $N(0,\mathrm{I})$ is 
\begin{align}
D(X) := D(f \| \phi) =\int f \log \frac{f}{\phi},
\end{align}
where $\phi (x) =  {(2\pi)^{-d/2}}e^{-|x|^2/2}$ denotes the   Gaussian density on $\mathbb{R}^d$.
By Jensen's inequality, $D(X)\geq 0$, with equality iff $X\sim N(0,\mathrm{I})$.  Similarly, the Fisher information of $X$ relative to $N(0,\mathrm{I})$ is defined according to 
\begin{align}
I(X) := I(f\|\phi)  = \int f \left| \nabla  \log \frac{f}{\phi}  \right|^2.
 \end{align}
As with $D(X)$, the quantity $I(X)$ is nonnegative, and zero only if $X\sim N(0,\mathrm{I})$.   Completely equivalent to \eqref{EPI} and \eqref{FII}, respectively, are the inequalities 
\begin{align}
\theta D(X) + \bar\theta D(Y) &\geq D(\sqrt{\theta}X + \sqrt{\bar\theta} Y)\label{relEPI}\\
\theta I(X) + ~\bar\theta I(Y) &\geq I(\sqrt{\theta}X + \sqrt{\bar\theta} Y),\label{relFII}
\end{align}
where $X,Y$ are independent zero-mean random vectors, $\theta\in[0,1]$ and $\bar{\theta}:=1-\theta$. Given their equivalence, we shall  refer to inequalities \eqref{EPI} and \eqref{relEPI} collectively as the {entropy power inequality} (EPI), and inequalities \eqref{FII} and \eqref{relFII} collectively as the Fisher information inequality (FII). 

Evidently, the EPI and FII apply separately to the entropy and Fisher information functionals.  However, in 1975, Gross established a remarkable inequality directly relating relative Fisher information to relative entropy \cite{gross1975logarithmic}.  Namely, 
\begin{align}
\dLSI(X) := \frac{1}{2}I(X) - D(X)\geq 0, \label{GrossLSI}
\end{align}
which is known as the logarithmic Sobolev inequality (LSI) for standard Gaussian measure.  The quantity $\dLSI(X)$ is the deficit in the LSI associated to $X$, and is zero iff $X$ is a translate of the standard normal \cite{carlen1991superadditivity}. The LSI has a variety of important consequences including Talagrand's quadratic transportation cost inequality \cite{talagrand1996transportation}, the Gaussian concentration inequality for Lipschitz functions (e.g., \cite{ledoux2005concentration}), and the Gaussian Poincar\'e inequality.  An interesting feature of the LSI (and many of its corollaries) is the fact that it  has no effective dependence on dimension. The convolution inequalities satisfied by Fisher information and relative entropy  also enjoy this dimension-free property. 

Although it was not recognized until the 1990s by Carlen \cite{carlen1991superadditivity}, Gross' LSI is in fact mathematically equivalent to the uncertainty principle
\begin{align}
\mathsf{p}(X) := \frac{1}{d}N(X) J(X) \geq 1,\label{stamIneq}
\end{align}
which was observed by Stam in his 1959 proof of \eqref{EPI} and \eqref{FII}.  We  refer to   the quantity $\mathsf{p}(X)$ as the \emph{Stam defect} associated to $X$, and remark here that  $\mathsf{p}(X)=1$ iff $X \sim N(0,\sigma^2 \mathrm{I})$ for some $\sigma^2>0$.   In fact, \eqref{stamIneq} is a direct consequence of the EPI and the so-called de Bruijn identity, which suggests a quantitative relationship between the LSI   and the EPI. Unfortunately, the derivation of the LSI  from the EPI does not  propagate any deficit terms, so only  conditions for equality are carried through.  In this paper, we fill this gap by proving a general inequality that interpolates between the LSI, the EPI and the FII. 

\subsection*{Organization}

The remainder of this paper is organized as follows.  Section \ref{sec:mainResults} contains our two main results, which are ultimately shown to be equivalent.  In particular, Section \ref{sec:InterpLSI} gives a general interpolation inequality for the LSI and EPI, followed by a brief discussion.  Section \label{sec:REPI} gives new reverse EPI and FII, and contains proofs of all main results. Section \ref{sec:Apps} provides  applications of the main results, including consequences for information-theoretic central limit theorems, stability of the LSI with respect to HWI jumps, and a sharp form of Nelson's hypercontractive inequality. In Section \ref{sec:Conc}, we conclude by speculating on the possibility of geometric analogues of our main results. 

\section{Main Results}\label{sec:mainResults}

\subsection{An Interpolation Inequality for the LSI and EPI}\label{sec:InterpLSI}

\begin{theorem}\label{thm:main}
Let $X,Y$ be independent, centered random vectors on $\mathbb{R}^d$.  For any $\theta \in[0,1]$ 
\begin{align}
D(X) + D(Y) \leq \frac{\bar \theta}{2}I(X) + \frac{ \theta}{2}I(Y) + D(\sqrt{\theta}X + \sqrt{\bar\theta} Y). \label{mainInequality}
\end{align}
 \end{theorem}
Let us briefly discuss a few observations.  If we adopt the convention that $\dLSI(X) = \infty$ when $I(X)=\infty$, and $\dLSI(X) =\tfrac{1}{2}I(X)-D(X)$ otherwise, then \eqref{mainInequality} may  be rewritten as
 \begin{align}
\theta D(X) + \bar \theta D(Y) \leq \bar \theta \dLSI(X) +  \theta \dLSI(Y)+ D(\sqrt{\theta}X + \sqrt{\bar\theta} Y),
 \end{align} 
 from which it is plain that \eqref{mainInequality} interpolates between the LSI for $X$, the LSI for $Y$ and the EPI.  However, Theorem \ref{thm:main} also directly connects the LSI to the FII.  Indeed,  definitions and  algebra yield the following equivalent form:
 \begin{align}
 \dLSI(\sqrt{\theta}X + \sqrt{\bar\theta} Y)+ \frac{ \theta}{2}I(X) + \frac{\bar  \theta}{2}I(Y)   \leq \dLSI(X) + \dLSI(Y) + \frac{1}{2} I(\sqrt{\theta}X + \sqrt{\bar\theta} Y),\label{mainFII}
 \end{align}
 which interpolates between the FII and the LSIs for $X$, $Y$ and the sum $\sqrt{\theta}X + \sqrt{\bar\theta} Y$.  Although it is a weakening of \eqref{mainFII}, we may apply the FII to  find that $\dLSI$ satisfies its own convolution inequality which remains an improvement of Gross' LSI:  
 \begin{corollary}[Convolution inequality for the LSI]\label{cor:convLSI}
 Let $X,Y$ be independent random vectors on $\mathbb{R}^d$.  For all $\theta\in[0,1]$
  \begin{align}
 \dLSI(\sqrt{\theta}X + \sqrt{\bar\theta} Y) \leq \dLSI(X) + \dLSI(Y).\label{convLSIineq}
 \end{align}
 \end{corollary}
\begin{remark}The assumption of centered random vectors is not needed  for \eqref{convLSIineq} due to the fact that   $\dLSI(\cdot)$ is translation invariant.  In fact, the centering assumption is not needed in Theorem \ref{thm:main} either, and can be dealt with by incorporating the inner product $\langle \EE X, \EE Y\rangle$ into \eqref{mainInequality}.  Details are straightforward and are left to the reader. 
\end{remark}
 
 Theorem \ref{thm:main} is essentially best possible for any choice of $X,Y$.  Indeed, observe that definitions and  the EPI  imply
 \begin{align}
 \dLSI(X) + \dLSI(Y) %
&\leq 2 \left( \frac{1}{4} I(X) +\frac{1}{4}I(Y)-\frac{1}{2}I\left(\tfrac{1}{\sqrt{2}}(X+Y)\right)  +  \dLSI\left(\tfrac{1}{\sqrt{2}}(X+Y)\right) \right).\label{sumDLSI_LB}
 \end{align}
 However, Theorem \ref{thm:main} implies via \eqref{mainFII} that 
 \begin{align}
  \dLSI(X) + \dLSI(Y) &\geq \sup_{\theta\in[0,1]} \left(  \frac{\theta}{2} I(X) +\frac{\bar \theta}{2}I(Y)-\frac{1}{2}I\left(\sqrt{\theta}X+\sqrt{\bar\theta}Y\right)  +  \dLSI\left(\sqrt{\theta}X+\sqrt{\bar\theta}Y\right)   \right) \\
  &\geq  \frac{1}{4} I(X) +\frac{1}{4}I(Y)-\frac{1}{2}I\left(\tfrac{1}{\sqrt{2}}(X+Y)\right)  +  \dLSI\left(\tfrac{1}{\sqrt{2}}(X+Y)\right), \label{sumDLSI_UB}
 \end{align}
 which differs from corresponding the upper bound \eqref{sumDLSI_LB} by precisely a factor of 2.  
 
Loosely speaking, the conjunction of \eqref{sumDLSI_LB} and \eqref{sumDLSI_UB}  suggests that $\dLSI(X)+\dLSI(Y)$  can be roughly decomposed into two nonnegative parts that depend jointly on $X,Y$: (i) the dissipation of Fisher information $\tfrac{1}{2}\left( I(X) + I(Y) \right) - I (\tfrac{1}{\sqrt 2}(X + Y ) )$; and (ii) the deficit in the LSI associated to the rescaled sum $\tfrac{1}{\sqrt{2}}(X+Y)$.   
On this note, we remark    that neither of these quantities  control one another in general.

To see this, let $\rho \in (0,1)$ and consider  Gaussian random vectors $X,Y$ with distributions:
\begin{align}
X\sim N\left(0 ,\begin{bmatrix} 1 & \rho \\ \rho & 1\end{bmatrix}\right) ~~~~~~ Y\sim N\left(0 ,\begin{bmatrix} 1 & -\rho \\ -\rho & 1\end{bmatrix}\right).
\end{align}
  In this case, $\tfrac{1}{\sqrt 2}(X+Y) \sim N(0,I)$ so $\delta_{\mathsf{LSI}}\left(\tfrac{1}{\sqrt 2}(X+ Y)\right)=0$.  However, 
\begin{align}
 \tfrac{1}{2}\left( I(X) + I(Y) \right) - I\left(\tfrac{1}{\sqrt 2}(X + Y )\right) = \frac{2}{1-\rho^2}  - 2 = \frac{2 \rho^2}{1-\rho^2}.
\end{align}
On the other hand, consider $X_*$ to be an independent copy of $X$.  In this case, $\tfrac{1}{\sqrt 2}(X + X_* )$ is equal to $X$ in distribution, so 
\begin{align}
 \tfrac{1}{2}\left( I(X) + I(X_*) \right) - I\left(\tfrac{1}{\sqrt 2}(X + X_* )\right) = 0.
\end{align}
However, we may readily compute that
\begin{align}
\delta_{\mathsf{LSI}}\left(\tfrac{1}{\sqrt 2}( X+ X_*)\right)=\frac{\rho^2}{1-\rho^2} + \frac{1}{2}\log (1-\rho^2) >\frac{\rho^2}{2}.
\end{align}

Another simple consequence of the above discussion is that, if $X,X_*$ are independent and identically distributed, then  
 \begin{align}
  \dLSI(X) \asymp   \frac{1}{2}\left( I(X) -I\left(\tfrac{1}{\sqrt{2}}(X+X_*)\right) \right) +  \dLSI\left(\tfrac{1}{\sqrt{2}}(X+X_*)\right), \label{decomposeDeficit}
 \end{align}
 where `$\asymp$' denotes equality up to an absolute constant factor.   This suggests that the Fisher information jump $I(X) -I(\tfrac{1}{\sqrt{2}}(X+X_*))$ can be used to quantify the stability of the LSI, a topic we will return to in Section \ref{sec:HWIjumps}.  
 
 Finally, we note that Theorem \ref{thm:main} allows us to easily deduce the (well-known) equality conditions for the LSI from those for the EPI.  Indeed, since $\dLSI(\cdot)$ and $D(\cdot)$ are invariant to unitary transformations, it follows from Theorem \ref{thm:main} that  $\dLSI(X)=\dLSI(Y)=0$ only if 
 \begin{align}
 \theta D(X) + \bar \theta D(Y) =  D(\sqrt{\theta}X + \sqrt{\bar\theta} \mathbf{U}Y), \label{EPIeq}
 \end{align}
 for all $\theta\in[0,1]$ and unitary matrices $\mathbf{U}: \mathbb{R}^d\to\mathbb{R}^d$. From conditions for equality in the EPI, this implies $\operatorname{Cov}(X)$ and $\operatorname{Cov}(Y)$ are  proportional to $\mathrm{I}$; in fact, direct computation shows they must be equal to satisfy \eqref{EPIeq}.  Thus, evaluation of  $\dLSI(X)$ for $X\sim N(\mu,\sigma^2\mathrm{I})$ allows us to conclude that $\dLSI(X)=0$ if and only if $X\sim N(\mu,\mathrm{I})$. 
 
 \subsection{Reverse Entropy Power and Fisher Information Inequalities} \label{sec:REPI}
 Due to its   fundamental  role in information theory, there has been sustained interest in obtaining reverse forms of the entropy power inequality \eqref{EPI}.  As shown by Bobkov and Chistyakov \cite{bobkov2015entropy}, the EPI cannot be reversed in general, at least not up to a constant factor.  Nevertheless, progress has been made.  A  notable example of a reverse EPI is due to Bobkov and Madiman \cite{bobkov2012reverse}, who show that for independent random vectors $X,Y$ with log-concave densities,  there exist linear volume preserving maps $u,v$ such that 
 \begin{align}
 N(u(X) + v(Y)) \leq C (N(X) + N(Y)),
 \end{align}
 where $C$ is an absolute constant.   Bobkov and Madiman's result mirrors Milman's reverse Brunn-Minkowski inequality \cite{milman1986inegalite}, which is pleasant since the EPI itself mirrors the Brunn-Minkowski inequality.  A similar statement holds for a more general class of convex measures. See also the recent survey  by Madiman, Melbourne and Xu \cite{madiman2016forward} for related results.  
 
 Another  example of a reverse EPI is due to  Ball, Nayar and Tkocz \cite{ball2015reverse}, who  restrict attention to the class of log-concave densities. They show that, for a symmetric log-concave vector $(X,Y)$ in $\mathbb{R}^2$, there is an absolute constant $\kappa$ such that 
 \begin{align}
 N(X+Y)^{\kappa} \leq N(X)^{\kappa} + N(Y)^{\kappa}. 
 \end{align}
 
We show below that both the entropy power inequality \eqref{EPI} and the Fisher information inequality \eqref{FII} can be precisely reversed, up to  factors that depend only on the Stam defects associated to $X$ and $Y$.  In particular, if both $X$ and $Y$ each nearly saturate Stam's inequality \eqref{stamIneq}, then the EPI and FII will also be nearly saturated.  Notably, strong regularity assumptions are not imposed.

\begin{theorem} \label{thm:REPI} Let $X,Y$ be independent random vectors on $\mathbb{R}^d$ with finite second moment, and choose $\lambda$   to satisfy $\lambda/(1-\lambda) = N(Y)/N(X)$.  Then
\begin{align}
N(X+Y) \leq \left( N(X)+N(Y) \right) \left( \lambda \mathsf{p}(X) + (1-\lambda) \mathsf{p}(Y)\right).\label{REPI}
\end{align}
Furthermore, if $J(X),J(Y)<\infty$, then 
\begin{align}
\frac{1}{J(X+Y)}\leq \left(\frac{1}{J(X)}+\frac{1}{J(Y)} \right)  \mathsf{p} (X)   \mathsf{p} (Y) .\label{RFII}
\end{align}
\end{theorem}

Since \eqref{REPI} does not immediately resemble either of the reverse inequalities mentioned above, let us briefly comment on how it may be understood in the context of known results.  In particular, it is well known that entropy power is concave under the action of the heat semigroup \cite{costa1985new, dembo1989simple, villani2000short}.  That is, if $G\sim N(0,\mathrm{I})$, then 
 \begin{align}
 \frac{\d^2}{\d t^2}N(X+ \sqrt{t} G)   \leq 0,\label{concaveEP}
 \end{align}
which is  the same as $N(X+ \sqrt{t} G)$ lying below its tangents lines.  By the semigroup property, it suffices to consider the tangent line at $t=0$, so \eqref{concaveEP} is equivalent to
\begin{align}
N(X+\sqrt{t}G) &\leq N(X) + t \left(\frac{\d}{\d t}N(X+\sqrt{t}G) \Big|_{t=0}\right)\\
&=  N(X) + t\, \mathsf{p}(X),
\end{align}
where the equality follows by de Bruijn's identity (e.g., \cite{stam1959some, bakry2013analysis, carlen1991entropy}). This coincides exactly with \eqref{REPI} particularized to the case where $Y = \sqrt{t} G$.   Thus, we may think of  \eqref{REPI} as a generalization of `concavity of entropy power' beyond the heat semigroup. 

Inequality \eqref{REPI} may also be viewed as a strengthening of Stam's uncertainty principle \eqref{stamIneq}.  Indeed, letting $X,X_*$ be IID, \eqref{REPI} reduces to
\begin{align}
\mathsf{p}(X) \geq \frac{N(\tfrac{1}{\sqrt 2} (X+X_*))}{N(X)}.
\end{align}
By the EPI, the term on the RHS is strictly greater than 1 unless $X$ is Gaussian.  On this note, we mention that Carlen and Soffer \cite{carlen1991entropy} have shown  that for $X$ centered with $\Cov(X)=\mathrm{I}$, there exists a nonnegative function $\Theta$ on $[0,\infty)$, strictly increasing from $0$ and depending only on certain decay and smoothness properties of $X$ such that
\begin{align}
 \frac{N(\tfrac{1}{\sqrt 2} (X+X_*))}{N(X)}\geq \exp\left(\frac{2}{d} \Theta (D(X)) \right).
\end{align}

Finally, we note that an equivalent version of Corollary \ref{cor:convLSI} for the Stam defect follows directly from \eqref{REPI} and the FII:
\begin{corollary}[Convolution inequality for the Stam defect] 
For $X,Y$ be independent random vectors on $\mathbb{R}^d$ with finite second moment,
\begin{align}
\mathsf{p}(X+Y) \leq \mathsf{p}(X)\mathsf{p}(Y). \label{subStamDefect}
\end{align}
\end{corollary}

The proof of Theorem \ref{thm:REPI} follows rather directly from a strengthening of the entropy power inequality proved recently by the author.  To state it, we first recall some notation familiar to information theorists:  Let $U,V$ have joint distribution $P_{UV}$ on the space $\mathcal{U}\times \mathcal{V}$ and let the respective marginals be denoted by $P_U,P_V$.  The mutual information $I(U;V)$ between $U$ and $V$ is given by
\begin{align}
I(U;V) := \EE \log \left( \frac{\d P_{UV}}{\d P_U \!\times \!P_V} \right) = \int_{\mathcal{U}\times \mathcal{V}}   \log \left( \frac{\d P_{UV}}{\d P_U \!\times \!P_V} \right) \d P_{UV}.
\end{align}
\begin{theorem} \cite{courtade2016strengthening}
Let $X,W$ be independent random vectors on $\mathbb{R}^d$, with $W$ Gaussian.  Define $Z=X+W$.  For any random variable $V$ such that $X,V$ are conditionally independent given $Z$, it holds that
\begin{align}
e^{-\tfrac{2}{d} I(X;V)}N(Z) \geq e^{-\tfrac{2}{d} I(Z;V)} N(X) + N(W).\label{sEPI}
\end{align}
\end{theorem}
\begin{proof}[Proof of Theorem \ref{thm:REPI}]
Let $Y$ be any random vector on $\mathbb{R}^d$, independent of $X,W$, and suppose $X,Y$ have finite second moments.  Then, $V = Z+Y$ is such that $X,V$ are conditionally independent given $Z$.  By the definition of mutual information, 
\begin{align}
I(X;V) &= h(X+Y+W) - h(Y+W)\\
I(Z;V) &= h(X+Y+W) - h(Y).
\end{align} 
Thus, rearranging exponents in \eqref{sEPI} gives the following
\begin{align}
N(X+W)N(Y+W) \geq N(X)N(Y) + N(X+Y+W) N(W).\label{3EPI}
\end{align}
Now, let $W = \sqrt{t}G$, where $G\sim N(0,\mathrm{I})$, in which case \eqref{3EPI} particularizes to 
\begin{align}
\frac{N(X+\sqrt{t}G)N(Y+\sqrt{t}G) - N(X)N(Y)}{t}\geq N(X+Y+\sqrt{t}G)  \geq N(X+Y).
\end{align}
An application of de Bruijn's identity $\frac{\d}{\d t}N(X+\sqrt{t}G) \Big|_{t=0}= \mathsf{p}(X)$ and the chain rule for derivatives proves 
\begin{align}
N(X+Y) &\leq N(X) \mathsf{p}(Y) +  N(Y) \mathsf{p}(X),
\end{align}
which is the same as \eqref{REPI}.  Stam's inequality $\mathsf{p}(X+Y)\geq1$ and algebra (valid when $J(X),J(Y)<\infty$) shows that \eqref{RFII} is a corollary of \eqref{REPI}.
\end{proof}

Theorem \ref{thm:main} now follows from Theorem \ref{thm:REPI}.  In fact: 
\begin{proposition}
Theorems \ref{thm:main} and \ref{thm:REPI} are equivalent. 
\end{proposition}
\begin{proof}
For convenience, we recall the scaling properties $N(t Z) = t^2 N(Z)$ and $t^2 J(t Z) = J(Z)$. Also, if $G_s\sim N(0, s\mathrm{I})$, the relative entropy $D(Z\|G_s)$ and Fisher information $I(Z\|G_s)$  are related to $h(Z)$ and $J(Z)$ via 
\begin{align}
 h\left( Z\right) -\frac{d}{2}\log(2 \pi e s) &= -D(Z\|G_s) + \frac{1}{2 s}\EE|Z|^2-\frac{d}{2}\label{relEnt_s}\\
J(Z) &= I(Z\|G_s) +\frac{2}{s}d - \frac{1}{s^2}\EE|Z|^2,\label{relFI_s}
\end{align} 
holding for any random vector $Z$ on $\mathbb{R}^d$ with $\EE|Z|^2<\infty$.  Further, $N(\cdot),J(\cdot)$ are translation invariant, so we may assume without loss of generality that all random vectors are centered.

$\bullet$ Proof of \eqref{REPI} $\Rightarrow$ \eqref{mainInequality}:   We assume $I(X),I(Y)<\infty$, else \eqref{mainInequality} is a tautology.  Now, finiteness of $I(X)$ implies $\EE|X|^2<\infty$, and similarly for $Y$ (see, e.g., \cite[Proof of Thm.~5]{carlen1991superadditivity}).   
Using the  scaling properties of $N(\cdot)$ and $J(\cdot)$, \eqref{REPI} implies
\begin{align}
N\left(  \sqrt{\theta}X+ \sqrt{\bar\theta}Y \right) \leq N(X) N(Y) \left( \frac{\bar \theta J( X) + \theta J(Y)}{d}\right),
\end{align}
Now, taking logarithms,  multiplying through by $d/2$ and recalling $\log x\leq x-1$, we have:
\begin{align}
h\left(  \sqrt{\theta}X+ \sqrt{\bar\theta}Y \right) -\frac{d}{2}\log(2 \pi e) &\leq h\left( X\right) -\frac{d}{2}\log(2 \pi e)+h\left(  Y \right) -\frac{d}{2}\log(2 \pi e)  \\
&\phantom{=}+  \frac{d}{2} \log\left( \frac{\bar \theta J( X) + \theta J(Y)}{d}\right)\notag \\
&\leq h\left( X\right) -\frac{d}{2}\log(2 \pi e)+h\left(  Y \right) -\frac{d}{2}\log(2 \pi e)  \\
&\phantom{\leq}+\frac{1}{2}\left(\bar \theta J(X)+\theta J(Y) \right)   -\frac{d}{2}.\notag
\end{align}
Now, \eqref{mainInequality} follows from the identities \eqref{relEnt_s} and \eqref{relFI_s} for $s=1$.

$\bullet$ Proof of \eqref{mainInequality} $\Rightarrow$ \eqref{REPI}: 
We may assume $X,Y$ have finite Fisher information and second moments.  With this assumption in place, consider any  $s>0$ and observe via straightforward manipulation using the identities \eqref{relEnt_s}-\eqref{relFI_s}  that \eqref{mainInequality} is equivalent to 
\begin{align}
D(X\|G_s) + D(Y\|G_s) \leq s \left( \frac{\bar\theta}{2} I(X\|G_s)+ \frac{\theta}{2}I(Y\|G_s)\right)  + D\left( \sqrt{\theta}X+ \sqrt{\bar\theta}Y \big \|G_s\right). \label{convLSIs}
\end{align}
 Hence, using \eqref{relEnt_s}-\eqref{relFI_s} again and rearranging, we find that this is the same as
\begin{align}
\log N\left( \sqrt{\theta}X+ \sqrt{\bar\theta}Y \right) \leq  \log N(X) N(Y)  + s \frac{\bar \theta J(X) + \theta J(Y)}{d} - \log s - 1.
\end{align}
Recalling $1 + \log a = \inf_{s>0} \left( as - \log s\right)$, we may minimize the RHS over $s>0$ to obtain
\begin{align}
d N\left( \sqrt{\theta}X+ \sqrt{\bar\theta}Y \right) &\leq N(X)N(Y)  \left( \bar \theta J(X) + \theta J(Y)\right)\\
&= N(\sqrt{\theta}X)N(\sqrt{\bar \theta}Y)  \left( J(\sqrt{\theta} X) +  J(\sqrt{\bar\theta} Y)\right),
\end{align}
where   the last equality follows via the scaling properties of $N(\cdot)$ and $J(\cdot)$.  A simple rescaling recovers \eqref{REPI} and completes the proof.
\end{proof}

\section{Applications} \label{sec:Apps}

\subsection{Short-term  convergence rates in information-theoretic CLTs}

Let $Z$ be a centered random vector on $\mathbb{R}^d$ with $\Cov(Z)=\mathrm{I}$, and define the normalized sums $U_n = \tfrac{1}{\sqrt n} \sum_{k=1}^n Z_k$, where $Z_1, Z_2, \dots, Z_n$ are independent copies of $Z$.  The entropic central limit theorem due to Barron \cite{barron1986entropy} asserts that $D(U_n)\to 0$, provided $D(U_{n_0})<\infty$  for some $n_0$.  Likewise, the CLT for Fisher information, due to Barron and Johnson \cite{johnson2004fisher}, asserts that $I(U_n)\to 0$, provided $I(U_{n_0})<\infty$  for some $n_0$.  

In 2004,  Artstein, Ball, Barthe and Naor established that each of these limit theorems enjoy monotone convergence \cite{artstein2004solution}. However, optimal estimates on the convergence rate remained open until recently.  On this front, Bobkov, Chistyakov, Gennadiy  and G{\"o}tze \cite{bobkov2013rate, bobkov2014berry} have  settled a longstanding conjecture and shown that under moment conditions
\begin{align}
D(U_n) = O(1/n), \label{Bobkov}
\end{align}
which is consistent with the convergence rates predicted by  the Berry-Esseen theorem. Although explicit constants are given for the $O(1/n)$ term as a function of the moments and  $D(Z)$, the  proof invokes local limit theorems for Edgeworth expansions, so the  $o(1/n)$ terms  are not explicitly quantified for finite $n$.  Hence, although \eqref{Bobkov} provides good long-term estimates on convergence in the entropic CLT, it does not immediately  provide any information about the short-term behavior of $D(U_n)$.    

The next result partially  addresses this issue by establishing a \emph{lower} bound on $D(U_n)$ in terms of $\dLSI(Z)$ and $n$.  Roughly speaking, if $\dLSI(Z) \ll D(Z)$, then $D(U_n)$ is assured to  decay slowly on short time scales.  A similar result holds for Fisher information.  That is, if $\dLSI(Z) \ll I(Z)$, then $I(U_n)$  will decay slowly on short time scales.  More precisely, each of these quantities decay at most linearly in $n$, with slope  $\dLSI(Z)$.

\begin{theorem}\label{thm:CLT}
Let $Z$ be a centered random vector on $\mathbb{R}^d$ and define the normalized sums $U_n = \tfrac{1}{\sqrt n} \sum_{k=1}^n Z_k$, where $Z_1, Z_2, \dots, Z_n$ are independent copies of $Z$. The sequence $\{ \delta_{\mathsf{LSI}}(U_n) , n\geq 1\}$ is subadditive.  Moreover, the  following hold for all $n\geq 1$:
\begin{align}
 D(U_n ) &\geq D(Z)  - (n-1) \delta_{\mathsf{LSI}}(Z) \label{entCLT}\\
 \frac{1}{2} I(U_n ) &\geq \frac{1}{2}  I(Z)  - n\,  \delta_{\mathsf{LSI}}(Z) + \delta_{\mathsf{LSI}}(U_n) .\label{fiCLT}
\end{align}

\end{theorem}
\begin{proof}
We apply Theorem \ref{thm:main} with $\theta=\frac{n}{n+m}$, $X = U_n$ and $Y = U_m$.  In this case, $\sqrt{\theta}X+ \sqrt{\bar\theta}Y$ is equal to $U_{n+m}$ in distribution, so we obtain the inequality
\begin{align}
D(U_{n+m}) \geq D(U_n) + D(U_m) - \frac{m}{2(m+n)} I(U_n) -  \frac{n}{2(m+n)} I(U_m),\label{entBase}
\end{align}
or, equivalently
\begin{align}
\delta_{\mathsf{LSI}}(U_{m}) +\delta_{\mathsf{LSI}}(U_{n})  \geq  \delta_{\mathsf{LSI}}(U_{m+n}) + \frac{n}{2(m+n)} I(U_n) +  \frac{m}{2(m+n)} I(U_m)-\frac{1}{2}I(U_{m+n}).\label{fiBase}
\end{align}
Subadditivity of $\delta_{\mathsf{LSI}}(U_{n})$ follows from applying \eqref{relFII}.

Now, the proof of \eqref{entCLT} and \eqref{fiCLT} will follow  by induction on $n+m$.   We first prove \eqref{entCLT}.  The base case for $n+m=2$ is immediate from \eqref{entBase} with $n=m=1$.  So, by induction, we have
\begin{align}
D(U_{n+m}) %
&\geq D(Z) - (n-1)  \delta_{\mathsf{LSI}}(Z)  + D(Z) -  (m-1) \delta_{\mathsf{LSI}}(Z) \\
&\phantom{\geq}- \frac{1}{2}\left( \frac{m}{m+n} I(U_n) + \frac{n}{m+n} I(U_m)\right)\notag \\
&= D(Z) - (n+m-1)  \delta_{\mathsf{LSI}}(Z) +\frac{1}{2}I(Z) - \frac{1}{2}\left( \frac{m}{m+n} I(U_n) +  \frac{n}{m+n} I(U_m)\right)\\ 
&\geq D(Z) - (n+m-1)  \delta_{\mathsf{LSI}}(Z),
\end{align}
where the final inequality is due to $I(Z) \geq I(U_m)$ for $m\geq 1$, a  consequence of \eqref{relFII}.

Now, we aim to prove \eqref{fiCLT}.  
The base case $n+m=2$ is immediate from \eqref{entBase} with $n=m=1$.  Thus, by the inductive hypothesis, we have
\begin{align}
&(m+n) \delta_{\mathsf{LSI}}(Z)\notag\\
&\geq \delta_{\mathsf{LSI}}(U_{m}) + \frac{1}{2}\left( I(Z|G)- I(U_{m}|G)\right)\\
&\phantom{\geq}+\delta_{\mathsf{LSI}}(U_{n}) + \frac{1}{2}\left( I(Z|G)- I(U_{n}|G)\right)\notag\\
&\geq \delta_{\mathsf{LSI}}(U_{m+n}) +   I(Z|G) -\left( \frac{m}{2(m+n)} I(U_n|G) +  \frac{n}{2(m+n)} I(U_m|G)+\frac{1}{2}I(U_{m+n}|G)\right)\label{fiIntermed}\\
&\geq \delta_{\mathsf{LSI}}(U_{m+n}) + \frac{1}{2}\left( I(Z|G)- I(U_{m+n}|G)\right),
\end{align}
where \eqref{fiIntermed} is \eqref{fiBase} and, as before,  the final inequality follows due to $I(Z) \geq I(U_m)$ for $m\geq 1$.
\end{proof}

\subsection{The LSI is stable with respect to HWI jumps}\label{sec:HWIjumps}

Recall the quadratic Wasserstein distance between random vectors $X,Y$ is defined according to
\begin{align}
W_2^2(X,Y) = \inf_{Q_{XY}\in \pi(X,Y)} \EE|X-Y|^2,
\end{align}
where the infimum is over all couplings between $X,Y$ that preserve their given marginals $X\sim P_X,Y\sim P_Y$. In the case where $G\sim N(0,\mathrm{I})$, Talagrand's quadratic transportation cost inequality asserts that
\begin{align}
W_2^2(X) := W_2^2(X,G) \leq 2D(X).\label{talagrand}
\end{align}
Talagrand's inequality is closely related to Gross' LSI.  Indeed, Otto and Villani \cite{otto2000generalization} proved a remarkable inequality that interpolates between \eqref{talagrand} and \eqref{GrossLSI}:
\begin{align}
D(X) \leq \sqrt{I(X)} W_2(X) - \frac{1}{2}W_2^2(X).  \label{HWI}
\end{align}
This inequality is referred to as the HWI inequality, since it simultaneously relates the relative entropy (H), Wasserstein distance (W), and  Fisher information (I) functionals. 

As with relative entropy and Fisher information, $W_2^2$ satisfies a convolution inequality (e.g., \cite{villani2003topics}).  That is, if $X,Y$ are independent, centered random vectors on $\mathbb{R}^d$ 
\begin{align}
W_2^2(\sqrt{\theta}X+\sqrt{\bar\theta} Y)\leq \theta W_2^2(X)+\bar\theta W_2^2(Y).
\end{align}
So, if $X,X^*$ are centered i.i.d.~random vectors on $\mathbb{R}^d$, then we have the following inequalities:
\begin{align}
D(\tfrac{1}{\sqrt{2}}(X+X_*)) &\leq D(X)\\
W_2(\tfrac{1}{\sqrt{2}}(X+X_*)) &\leq W_2(X)\\
I(\tfrac{1}{\sqrt{2}}(X+X_*)) &\leq I(X).
\end{align}
We shall use the term \emph{HWI jump} to refer to a deficit in any of the three inequalities above.  An immediate consequence of Theorem \ref{thm:main} and \eqref{HWI} is that the LSI is stable with respect to HWI jumps.  That is, if $X$ exhibits a jump under convolution with respect to any of the HWI functionals, then 
we can plainly lower bound the deficit in the LSI in terms of $I(X)$.

\begin{theorem}\label{thm:HWIjump}
Let $X,X^*$ be centered i.i.d.~random vectors on $\mathbb{R}^d$. If any of the following hold:
\begin{enumerate}[(i)]
\item $D(\tfrac{1}{\sqrt{2}}(X+X_*))\leq (1-\varepsilon)D(X)$, or
\item $W_2(\tfrac{1}{\sqrt{2}}(X+X_*))\leq (1-\varepsilon^{1/2})W_2(X)$, or
\item $I(\tfrac{1}{\sqrt{2}}(X+X_*))\leq (1-\varepsilon)I(X)$, 
\end{enumerate}
then $\delta_{\mathsf{LSI}}(X)\geq \tfrac{\varepsilon}{4}I(X).$
\end{theorem} 

Before proving Theorem \ref{thm:HWIjump}, we remark  that there have been a number of recent attempts to give quantitative stability estimates for the LSI.  However, such stability estimates are generally dimension dependent \cite{bobkov2014bounds} or impose strong regularity conditions such as presence of a spectral gap \cite{fathi2014quantitative, indrei2013quantitative}.  In contrast, Theorem \ref{thm:HWIjump} shows that HWI jumps give  a dimension-free estimate of the deficit in the LSI that holds without stringent regularity assumptions. However, HWI jumps do not directly bound the distance of $X$ from normal except under regularity conditions such as presence of a spectral gap ($d=1$) \cite{ball2003entropy} and log-concave   density ($d\geq 2$) \cite{ball2012entropy},  or a radial symmetry assumption ($d\geq 2$) \cite{courtade2016jumps}.   Of course, one cannot hope that $\delta_{\mathsf{LSI}}(X)\geq \varepsilon I(X)$ in general for some absolute constant $\varepsilon$, else it would contradict optimality of the constant   in Gross' LSI.

\begin{proof}[Proof of Theorem \ref{thm:HWIjump}]
The proof is an immediate consequence of Theorem \ref{thm:main} and the HWI inequality, but we include it for completeness. In all cases, we shall apply Theorem \ref{thm:main} with $Y=X_*$ and $\theta=1/2$.   In view of this, if (i) holds, then Theorem \ref{thm:main} implies
\begin{align}
D(X) \leq \frac{1}{2(1+\varepsilon)}I(X) = \frac{1}{2(1+\varepsilon)}I(X)\leq \frac{1}{2}I(X) - \frac{\varepsilon}{4}I(X).
\end{align}

Next, if (ii) holds, then using the HWI inequality and the convolution inequality for Fisher information, we have:
\begin{align}
2 D(X)  &\leq \frac{1}{2} I(X)+ D\left( \tfrac{1}{\sqrt{2}}(X+X_*) \right)\\
&\leq \frac{1}{2} I(X)  + \sqrt{   I(\tfrac{1}{\sqrt{2}}(X+X_*))   } W_2(\tfrac{1}{\sqrt{2}}(X+X_*)) - \frac{1}{2}W_2^2(\tfrac{1}{\sqrt{2}}(X+X_*))\\
&\leq \frac{1}{2} I(X)  + \sqrt{   I(X)   } W_2(\tfrac{1}{\sqrt{2}}(X+X_*)) - \frac{1}{2}W_2^2(\tfrac{1}{\sqrt{2}}(X+X_*)).
\end{align}
Equivalently, 
\begin{align}
 \frac{1}{2} I(X)  \leq 2\dLSI(X)  +  \sqrt{  I(X)   } W_2(\tfrac{1}{\sqrt{2}}(X+X_*)) - \frac{1}{2}W_2^2(\tfrac{1}{\sqrt{2}}(X+X_*)).
\end{align}
By the conjunction of Talagrand's inequality and the LSI,  (ii) and the quadratic formula, we can conclude:
\begin{align}
(1-\varepsilon^{1/2})^2  I(X) 
&\geq
(1-\varepsilon^{1/2})^2 \theta W_2^2(X) \\
&\geq
W_2^2(\tfrac{1}{\sqrt{2}}(X+X_*)) \\
 &\geq \left(  \sqrt{  I(X) } - 2 \sqrt{ \dLSI(X)    }  \right)^2.
\end{align}
Taking square roots, rearranging and squaring again, we obtain the desired inequality.

Finally, if (iii) holds, then the claim is immediate from \eqref{mainFII}.
\end{proof}

\subsection{A sharp form of Nelson's hypercontractivity estimate}

It is well known that Gross' LSI is equivalent to Nelson's hypercontractivity estimate for the Ornstein-Uhlenbeck semigroup \cite{gross1975logarithmic}.  To state Nelson's result, let us first introduce the Ornstein-Uhlenbeck semigroup $(P_t)_{t\geq 0}$ defined on functions $f: \mathbb{R}^d \to \mathbb{R}$ as follows:
\begin{align}
P_t f(x) = \int_{\mathbb{R}^d} f( e^{-t} x + (1-e^{-2t})^{1/2}y) \d \gamma(y),
\end{align}
where $\gamma$ denotes the standard Gaussian measure on $\mathbb{R}^d$.  Nelson's result is as follows:
\begin{theorem}  \cite{nelson1973free}
 For   $f\in L^p(\gamma)$
\begin{align}
\|P_t f\|_{L^q(\gamma)}     \leq  \| f\|_{L^p(\gamma)}   \label{NelsonHC}
\end{align}
for all  $q \geq p >1$ such that $q\leq 1 + (p-1)e^{2t}$.
\end{theorem}

The essential idea  behind  Gross' proof of \eqref{NelsonHC} from the LSI is  as follows: Let $q(t) = 1 + (p-1)e^{2t}$.  If $f\geq 0$ is a  smooth function, then $\| P_t f\|_{L^{q(t)}(\gamma)}$ is differentiable on $t\in[0,\infty)$ with derivative
\begin{align}
 \frac{\d}{\d t} \log \left( \| P_t f\|_{L^{q(t)}(\gamma)} \right) %
&=\frac{q'(t)}{q^2(t)} D(X_t) -   \frac{2(q(t) - 1)}{q^2(t)     } \frac{1}{2}I(X_t),
\end{align}
where $X_t$ is the random variable having  density (with respect to $\gamma$) proportional to $(P_t f)^{q(t)}$.  On account of the fact that $q'(t) = 2(q(t)-1)$, Gross' LSI implies
\begin{align}
\frac{\d}{\d t} \| P_t f\|_{L^{q(t)}(\gamma)} \leq 0.
\end{align}
Since $\left. \| P_t f\|_{L^{q(t)}(\gamma)} \right|_{t=0} = \| f\|_{L^p(\gamma)}$, Nelson's result follows for smooth $f$.  The extension to $f \in L^p(\gamma)$ follows by a density argument. The reverse implication is also apparent.  That is, in order for \eqref{NelsonHC} to hold for smooth $f\geq 0$, we must have $ \frac{\d}{\d t} \| P_t f\|_{L^{q(t)}(\gamma)} \leq 0$ at $t=0$, which is precisely Gross' LSI for the random variable $X$  having  density (with respect to $\gamma$) proportional to $f^p$. Indeed, by the semigroup property, Nelson's inequality is completely characterized by the behavior of $\| P_t f\|_{L^{q(t)}(\gamma)}$ in a neighborhood of $t=0$. 

Extending this to the setting of Theorem \ref{thm:main} is straightforward. Before we state the extension, let us introduce some notation.  For  $f,g\in L^p(\gamma)$, let $X$ be the random vector having density (with respect to $\gamma$) proportional to $|f|^p$, and let $Y$ be the random vector having density (with respect to $\gamma$) proportional to $|g|^p$.  Further, define their centered counterparts $\hat X = X-\EE X$ and $\hat Y =Y-\EE Y$, and the associated entropy production functional:
\begin{align}
E_{p,\theta}(f,g) := \theta D(\hat X) + \bar\theta D(\hat Y) - D(\sqrt{\theta} \hat X + \sqrt{\bar \theta } \hat Y).
\end{align}
We have the following improvement to Nelson's result, which interpolates between the hypercontractive estimates for two functions $f,g\in L^p(\gamma)$:
\begin{theorem} \label{thm:NelsonNew}
Let $p'$ denote the H\"older conjugate of $p$.  For smooth functions $f,g\in L^p(\gamma)$, 
\begin{align}
 \| P_t f\|^{\bar \theta }_{L^{q}(\gamma)} \| P_t g\|^{  \theta }_{L^{q}(\gamma)}  \leq \exp\left( - \frac{2\, t }{p \,p'\,} E_{p,\theta}(f,g) + o(t)   \right)  \| f\|_{L^p(\gamma)}^{\bar \theta} \|  g\|_{L^p(\gamma)}^{  \theta} 
\end{align}
for all  $q \geq p >1\geq \theta \geq 0$ such that $q\leq 1 + (p-1)e^{2t}$. 
\end{theorem}
\begin{proof}
We proceed in a manner identical to Gross' argument outlined above.  In particular, we may assume $f,g\geq 0$.  Then, using the definitions of $X,Y$ implicit in the definition of $E_{p,\theta}(f,g)$, we have
\begin{align}
\left. \frac{\d}{\d t} \log \left( \| P_t f\|^{\bar \theta }_{L^{q(t)}(\gamma)} \| P_t g\|^{  \theta }_{L^{q(t)}(\gamma)} \right) \right|_{t=0}
&=- \bar \theta \frac{q'(0)}{q^2(0)} \dLSI(X)   -   \theta \frac{q'(0)}{q^2(0)}  \dLSI(Y)\\
&\leq - \frac{2(p-1)}{p^2} \left( \theta D(\hat X) + \bar\theta D(\hat Y) - D(\sqrt{\theta} \hat X + \sqrt{\bar \theta } \hat Y) \right)\\
&= - \frac{2(p-1)}{p^2} E_{p,\theta}(f,g).
\end{align}
where the  inequality follows from Theorem \ref{thm:main}, translation invariance of $\dLSI(\cdot)$ and definition of $q(t)$. Thus, since $\| P_t f\|^{\bar \theta }_{L^{q(t)}(\gamma)} \| P_t g\|^{  \theta }_{L^{q(t)}(\gamma)}$ is differentiable in $t\in[0,\infty)$ by the smoothness assumption, it follows that 
\begin{align}
\log \left( \| P_t f\|^{\bar \theta }_{L^{q(t)}(\gamma)} \| P_t g\|^{  \theta }_{L^{q(t)}(\gamma)} \right) \leq \log\left(  \| f\|_{L^p(\gamma)}^{\bar \theta} \|  g\|_{L^p(\gamma)}^{  \theta}\right) - t \frac{2(p-1)}{p^2} E_{p,\theta}(f,g) + o(t) 
\end{align}
as desired. 
\end{proof}
\begin{remark}
The smoothness assumption on $f,g$ cannot be immediately eliminated in Theorem \ref{thm:NelsonNew} by a density argument.  Indeed, smooth functions may be dense in $L^p(\gamma)$, but the functional $E_{p,\theta}(f,g)$ is not necessarily (semi-)continuous in its arguments with respect to the $L^p(\gamma)$-norm, since relative entropy is only weakly lower semicontinuous in general.  %
\end{remark}

\section{Concluding Remarks}\label{sec:Conc}

To close, we briefly speculate on the potential for geometric analogues of Theorem \ref{thm:REPI}.  The EPI is often compared to the Brunn-Minkowski inequality, which states that for two nonempty compact subsets $A,B$ of $\mathbb{R}^d$, 
\begin{align}
 \operatorname{Vol}  (A+B)^{1/d}\geq  \operatorname{Vol}  (A)^{1/d}+ \operatorname{Vol} (B)^{1/d},
\end{align}
where $A + B$ denotes the Minkowski sum $A+B:=\{\,a+b\in \mathbb {R} ^{d}\mid a\in A,\ b\in B\,\}$.  Indeed, if $X$ is uniformly distributed on $A$, then $2\pi e \,N(X) =  \operatorname{Vol}  (A)^{2/d}$.  As already seen in Section \ref{sec:REPI}, the reverse EPI due to Bobkov and Madiman compares similarly with   Milman's reverse {B}runn--{M}inkowski inequality, which  holds when $A,B$ are convex.  

The classical isoperimetric inequality states that, for a sufficiently regular  subset  $A $ of $\mathbb{R}^d$, the surface area $|\partial A|$ exceeds that of $\mathcal{B}_A$, a ball in $\mathbb{R}^d$ with the same volume as $A$.  Notably, the isoperimetric inequality can be derived from the Brunn-Minkowski inequality in a manner very similar to how Stam's inequality \eqref{stamIneq}  follows from the EPI (e.g., \cite{costa1984similarity}).  Thus, there is a strong analogy between these information-theoretic inequalities and their geometric counterparts.  

In light of this, we are moved to speculate that a geometric analogue to Theorem \ref{thm:main} may hold.  For example, it seems reasonable to posit the following for sufficiently regular $A,B\subset\mathbb{R}^d$:
\begin{align}
  \operatorname{Vol}  (A+B)^{1/d} \leq  \left( \operatorname{Vol}  (A)^{1/d}+ \operatorname{Vol} (B)^{1/d}\right) \left( \lambda \frac{|\partial A |}{|\partial \mathcal{B}_A |} + \bar{\lambda} \frac{|\partial B |}{|\partial \mathcal{B}_B |} \right),
\end{align}
where $\lambda$ is chosen to satisfy $\lambda/(1-\lambda) =  \operatorname{Vol}  (B)^{1/d}  /  \operatorname{Vol}  (A)^{1/d}$.  As of now, we have made no  attempt to prove or disprove this inequality.

Finally, analogous to \eqref{subStamDefect}, we might expect to find that, for sufficiently regular $A,B\subset\mathbb{R}^d$
\begin{align}
\frac{|\partial (A+B) |}{|\partial \mathcal{B}_{A+B} |} \leq   \frac{|\partial A |}{|\partial \mathcal{B}_A |}   \frac{|\partial B |}{|\partial \mathcal{B}_B |} .
\end{align}
Again, we have made no  attempt to prove or disprove this inequality.

\subsection*{Acknowledgment}
This work was supported in part by NSF grants CCF-1528132 and CCF-0939370 (Center for Science of Information).

\bibliographystyle{IEEEtran}
\bibliography{jumpsBib}

\begin{thebibliography}{10}
\providecommand{\url}[1]{#1}
\csname url@samestyle\endcsname
\providecommand{\newblock}{\relax}
\providecommand{\bibinfo}[2]{#2}
\providecommand{\BIBentrySTDinterwordspacing}{\spaceskip=0pt\relax}
\providecommand{\BIBentryALTinterwordstretchfactor}{4}
\providecommand{\BIBentryALTinterwordspacing}{\spaceskip=\fontdimen2\font plus
\BIBentryALTinterwordstretchfactor\fontdimen3\font minus
  \fontdimen4\font\relax}
\providecommand{\BIBforeignlanguage}[2]{{%
\expandafter\ifx\csname l@#1\endcsname\relax
\typeout{** WARNING: IEEEtran.bst: No hyphenation pattern has been}%
\typeout{** loaded for the language `#1'. Using the pattern for}%
\typeout{** the default language instead.}%
\else
\language=\csname l@#1\endcsname
\fi
#2}}
\providecommand{\BIBdecl}{\relax}
\BIBdecl

\bibitem{stam1959some}
A.~Stam, ``Some inequalities satisfied by the quantities of information of
  {F}isher and {S}hannon,'' \emph{Information and Control}, vol.~2, no.~2, pp.
  101--112, 1959.

\bibitem{blachman1965convolution}
N.~M. Blachman, ``The convolution inequality for entropy powers,''
  \emph{Information Theory, IEEE Transactions on}, vol.~11, no.~2, pp.
  267--271, 1965.

\bibitem{gross1975logarithmic}
L.~Gross, ``Logarithmic {S}obolev inequalities,'' \emph{American Journal of
  Mathematics}, vol.~97, no.~4, pp. 1061--1083, 1975.

\bibitem{carlen1991superadditivity}
E.~A. Carlen, ``Superadditivity of {F}isher's information and logarithmic
  {S}obolev inequalities,'' \emph{Journal of Functional Analysis}, vol. 101,
  no.~1, pp. 194--211, 1991.

\bibitem{talagrand1996transportation}
M.~Talagrand, ``Transportation cost for {G}aussian and other product
  measures,'' \emph{Geometric \& Functional Analysis GAFA}, vol.~6, no.~3, pp.
  587--600, 1996.

\bibitem{ledoux2005concentration}
M.~Ledoux, \emph{The concentration of measure phenomenon}.\hskip 1em plus 0.5em
  minus 0.4em\relax American Mathematical Soc., 2005, no.~89.

\bibitem{bobkov2015entropy}
S.~G. Bobkov and G.~P. Chistyakov, ``Entropy power inequality for the
  {R}{\'e}nyi entropy,'' \emph{IEEE Transactions on Information Theory},
  vol.~61, no.~2, pp. 708--714, 2015.

\bibitem{bobkov2012reverse}
S.~Bobkov and M.~Madiman, ``Reverse {B}runn--{M}inkowski and reverse entropy
  power inequalities for convex measures,'' \emph{Journal of Functional
  Analysis}, vol. 262, no.~7, pp. 3309--3339, 2012.

\bibitem{milman1986inegalite}
V.~D. Milman, ``In{\'e}galit{\'e} de {B}runn-{M}inkowski inverse et
  applications \`a la th{\'e}orie locale des espaces norm{\'e}s,'' \emph{CR
  Acad. Sci. Paris}, vol. 302, no.~1, pp. 25--28, 1986.

\bibitem{madiman2016forward}
M.~Madiman, J.~Melbourne, and P.~Xu, ``Forward and reverse entropy power
  inequalities in convex geometry,'' \emph{arXiv preprint arXiv:1604.04225},
  2016.

\bibitem{ball2015reverse}
K.~Ball, P.~Nayar, and T.~Tkocz, ``A reverse entropy power inequality for
  log-concave random vectors,'' \emph{arXiv preprint arXiv:1509.05926}, 2015.

\bibitem{costa1985new}
M.~Costa, ``A new entropy power inequality,'' \emph{IEEE Transactions on
  Information Theory}, vol.~31, no.~6, pp. 751--760, 1985.

\bibitem{dembo1989simple}
A.~Dembo, ``Simple proof of the concavity of the entropy power with respect to
  added gaussian noise,'' \emph{IEEE Transactions on Information Theory},
  vol.~35, no.~4, pp. 887--888, 1989.

\bibitem{villani2000short}
C.~Villani, ``A short proof of the ?concavity of entropy power?'' \emph{IEEE
  Transactions on Information Theory}, vol.~46, no.~4, pp. 1695--1696, 2000.

\bibitem{bakry2013analysis}
D.~Bakry, I.~Gentil, and M.~Ledoux, \emph{Analysis and geometry of Markov
  diffusion operators}.\hskip 1em plus 0.5em minus 0.4em\relax Springer Science
  \& Business Media, 2013, vol. 348.

\bibitem{carlen1991entropy}
E.~A. Carlen and A.~Soffer, ``Entropy production by block variable summation
  and central limit theorems,'' \emph{Communications in mathematical physics},
  vol. 140, no.~2, pp. 339--371, 1991.

\bibitem{courtade2016strengthening}
T.~A. Courtade, ``Strengthening the entropy power inequality,'' \emph{arXiv
  preprint arXiv:1602.03033}, 2016.

\bibitem{barron1986entropy}
A.~R. Barron, ``Entropy and the central limit theorem,'' \emph{The Annals of
  probability}, pp. 336--342, 1986.

\bibitem{johnson2004fisher}
O.~Johnson and A.~Barron, ``Fisher information inequalities and the central
  limit theorem,'' \emph{Probability Theory and Related Fields}, vol. 129,
  no.~3, pp. 391--409, 2004.

\bibitem{artstein2004solution}
S.~Artstein, K.~Ball, F.~Barthe, and A.~Naor, ``Solution of shannon?s problem
  on the monotonicity of entropy,'' \emph{Journal of the American Mathematical
  Society}, vol.~17, no.~4, pp. 975--982, 2004.

\bibitem{bobkov2013rate}
S.~G. Bobkov, G.~P. Chistyakov, and F.~G{\"o}tze, ``Rate of convergence and
  edgeworth-type expansion in the entropic central limit theorem,'' \emph{The
  Annals of Probability}, vol.~41, no.~4, pp. 2479--2512, 2013.

\bibitem{bobkov2014berry}
------, ``Berry-esseen bounds in the entropic central limit theorem,''
  \emph{Probability Theory and Related Fields}, vol. 159, no. 3-4, p. 435,
  2014.

\bibitem{otto2000generalization}
F.~Otto and C.~Villani, ``Generalization of an inequality by {T}alagrand and
  links with the logarithmic {S}obolev inequality,'' \emph{Journal of
  Functional Analysis}, vol. 173, no.~2, pp. 361--400, 2000.

\bibitem{villani2003topics}
C.~Villani, \emph{Topics in optimal transportation}.\hskip 1em plus 0.5em minus
  0.4em\relax American Mathematical Soc., 2003, no.~58.

\bibitem{bobkov2014bounds}
S.~G. Bobkov, N.~Gozlan, C.~Roberto, and P.-M. Samson, ``Bounds on the deficit
  in the logarithmic {S}obolev inequality,'' \emph{Journal of Functional
  Analysis}, vol. 267, no.~11, pp. 4110--4138, 2014.

\bibitem{fathi2014quantitative}
M.~Fathi, E.~Indrei, and M.~Ledoux, ``Quantitative logarithmic {S}obolev
  inequalities and stability estimates,'' \emph{arXiv preprint
  arXiv:1410.6922}, 2014.

\bibitem{indrei2013quantitative}
E.~Indrei and D.~Marcon, ``A quantitative log-{S}obolev inequality for a two
  parameter family of functions,'' \emph{International Mathematics Research
  Notices}, p. rnt138, 2013.

\bibitem{ball2003entropy}
K.~Ball, F.~Barthe, and A.~Naor, ``Entropy jumps in the presence of a spectral
  gap,'' \emph{Duke Mathematical Journal}, vol. 119, no.~1, pp. 41--63, 2003.

\bibitem{ball2012entropy}
K.~Ball and V.~H. Nguyen, ``Entropy jumps for isotropic log-concave random
  vectors and spectral gap,'' \emph{Studia Mathematica}, vol. 213, no.~1, pp.
  81--96, 2012.

\bibitem{courtade2016jumps}
T.~A. Courtade, ``Entropy jumps for radially symmetric random vectors,''
  \emph{preprint}, 2016.

\bibitem{nelson1973free}
E.~Nelson, ``The free {M}arkoff field,'' \emph{Journal of Functional Analysis},
  vol.~12, no.~2, pp. 211--227, 1973.

\bibitem{costa1984similarity}
M.~Costa and T.~Cover, ``On the similarity of the entropy power inequality and
  the {B}runn-{M}inkowski inequality (corresp.),'' \emph{IEEE Transactions on
  Information Theory}, vol.~30, no.~6, pp. 837--839, 1984.

\end{thebibliography}

\end{document}